\documentclass[runningheads]{llncs}
\usepackage[utf8]{inputenc}
\usepackage[T1]{fontenc}

\usepackage{graphicx}

\usepackage{amssymb}
\usepackage{amsmath}
\usepackage{xspace}
\usepackage{nicefrac}

\usepackage{pgfplots}
\pgfplotsset{compat=1.16}
\usetikzlibrary{math}
\usepgfplotslibrary{fillbetween}
\newtheorem{theo}{Theorem}
\newtheorem{prop}[theo]{Proposition}
\newtheorem{lemm}[theo]{Lemma}

\newtheorem{myclaim}{Claim}

\newcommand{\ALG}{\texttt{ALG}\xspace}
\newcommand{\OPT}{\texttt{OPT}\xspace}
\newcommand{\KB}{\texttt{KB}\xspace}
\newcommand{\NF}{\texttt{NF}\xspace}
\newcommand{\FF}{\texttt{FF}\xspace}

\newcommand{\WNFD}{\texttt{WNFD}\xspace}
\newcommand{\WNFI}{\texttt{WNFI}\xspace}
\newcommand{\WFFD}{\texttt{WFFD}\xspace}
\newcommand{\WFFI}{\texttt{WFFI}\xspace}
\newcommand{\WNFDR}{\texttt{WNFD-R}\xspace}
\newcommand{\WFFDR}{\texttt{WFFD-R}\xspace}
\newcommand{\WFFIR}{\texttt{WFFI-R}\xspace}

\newcommand{\cqfd} {\mbox{}~\hfill{\lower-0.3ex\hbox{\tiny $\blacksquare$}}}

\newcommand{\Diag}{\operatorname{Diag}}
\newcommand{\R}{\mathbb{R}}

\newcommand{\nb}{\boldsymbol{n}}
\newcommand{\bb}{\boldsymbol{b}}
\newcommand{\xb}{\boldsymbol{x}}
\newcommand{\yb}{\boldsymbol{y}}
\newcommand{\zb}{\boldsymbol{z}}
\newcommand{\Lb}{\boldsymbol{L}}
\newcommand{\vb}{\boldsymbol{v}}
\newcommand{\wb}{\boldsymbol{w}}
\newcommand{\ab}{\boldsymbol{a}}
\newcommand{\ub}{\boldsymbol{u}}
\newcommand{\Vb}{\boldsymbol{V}}
\newcommand{\Db}{\boldsymbol{D}}
\newcommand{\Ub}{\boldsymbol{U}}
\newcommand{\mb}{\boldsymbol{m}}
\newcommand{\sqts}{\sqrt{37}}

\newcommand{\bB}{\overline{B}}

\newcommand{\bO}{\overline{O}}

\newcommand\scalemath[2]{\scalebox{#1}{\mbox{\ensuremath{\displaystyle #2}}}}

\usepackage{xcolor}

\begin{document}

\title{Improved Analysis of two Algorithms for Min-Weighted Sum Bin Packing}
\author{Guillaume Sagnol
\orcidID{0000-0001-6910-8907} }
\authorrunning{G. Sagnol}
\titlerunning{Improved Analysis of Two Algorithms for MWSBP}

\institute{Technische Universität Berlin, 
Institut für Mathematik\\ 
Straße des 17. Juni 136, 10623 Berlin, Germany\quad
\email{sagnol@math.tu-berlin.de}\\
}

\maketitle

\begin{abstract}
We study the Min-Weighted Sum Bin Packing problem, a variant of the classical Bin Packing problem in which items have a weight, and each item induces a cost equal to its weight multiplied by the index of the bin in which it is packed.
This is in fact equivalent to a batch scheduling problem that arises in many fields of applications such as appointment scheduling or warehouse logistics. We give improved lower and upper bounds on the approximation ratio of two simple algorithms for this problem. In particular, we show that the knapsack-batching algorithm, which iteratively solves knapsack problems over the set of remaining items to pack the maximal weight in the current bin, has an approximation ratio of at most 17/10.

\vspace{-0.3em}
\end{abstract}

\section{Introduction} \vspace{-0.3em}\enlargethispage{1em}
\textsc{Bin Packing} is a fundamental problem in computer science, in which a set of items must be packed into the smallest possible number of identical bins, and has applications in fields as varied as logistics, data storage or cloud computing. A property of the bin packing objective is that all bins are treated as ``equally good'', which is not  always true in applications with a temporal component. Consider, e.g., the problem of allocating a set of $n$ patients to days for a medical appointment with a physician. Each patient is characterized by a service time and a weight indicating the severity of his health condition. The total time required to examine all patients assigned to a given day should not exceed a fixed threshold. The days thus correspond to bins indexed by $1,2,\ldots$, and bins with a small index are to be favored, especially for patients with a high weight, because they yield a smaller waiting time for the patients.

In the \textsc{Min-Weighted Sum Bin Packing Problem} (MWSBP), which was formally introduced in~\cite{EpsteinL07a}, the input consists of a set of $n$ items with size $s_i\in(0,1]$ and weight $w_i>0$\footnote{In~\cite{EpsteinL07a} the weights are assumed
to be $w_i\geq 1$, but we can reduce to this case by scaling; Our
lower bounds on approximation ratios are not affected by this operation.}. The goal is to find a feasible allocation of minimum cost of the set of items to bins, i.e., a partition of $[n]:=\{1,\ldots,n\}$ into $B_1,\ldots,B_p$ such that $\sum_{i\in B_k} s_i \leq 1$ holds for all $k\in[p]$, where the cost of putting item $i$ into $B_k$ is given by $k\cdot w_i$. In other words, 
if we use the notation $x(S):=\sum_{i\in S} x_i$ for a vector $x\in\R^n$ and a subset $S\subseteq[n]$,
the cost of a feasible allocation is
\begin{equation}
 \Phi(B_1,\ldots,B_p):=\sum_{k=1}^p k\cdot w(B_k)
 =
\sum_{k=1}^p \sum_{j=k}^p w(B_j). \label{Phi0}
\end{equation}

Another interpretation of this problem is that we have a set of jobs with unit processing times, and want to schedule them on a batch processing machine capable of simultaneously processing a batch of jobs of total size at most $1$, with the objective to minimize the weighted sum of completion times. 
In the three-field notation introduced by Graham et al.~\cite{GrahamLLR79}, MWSBP is thus equivalent to $1|p\operatorname{-}batch, s_j, p_j=1|\sum w_j C_j$; we refer to~\cite{FowlerM22} for a recent review giving more background on parallel batch scheduling.
Broadly speaking, we see that MWSBP captures the main challenge of many real-world problems in which items must be packed over time, such as the appointment scheduling problem mentioned above, or the problem of 
scheduling batches of orders in an order-picking warehouse.

For a given instance of MWSBP and an algorithm \ALG, we denote by \OPT the cost of an optimal solution, and the cost of the solution returned by \ALG is denoted by \ALG as well. Recall that the approximation ratio $\mathcal{R}(\ALG)$ of an algorithm
$\ALG$ is the smallest constant $\rho\geq 1$ such that, for all instances of the problem, it holds $\ALG \leq \rho \cdot \OPT$.

\enlargethispage{1em} \vspace{-0.4em}
\paragraph{Related Work.} 
The complexity of MWSBP is well understood, as the problem is NP-hard in the strong sense and a polynomial-time approximation scheme (PTAS) exists~\cite{EpsteinL07a}.
This paper also gives a simple algorithm based on \texttt{Next-Fit} 
which has an approximation ratio of 2. Prior to this work, 
several heuristics have been proposed 
for a generalization of the problem with incompatible families of items associated with different batch times~\cite{DobsonN01},
and it was shown that a variant of \texttt{First-Fit}
has an approximation ratio of 2 as well.

The unweighted version of the problem, \texttt{Min-Sum Bin Packing} (MSBP), in which each item has weight $w_i=1$, also attracted some attention. A simpler PTAS is described in~\cite{EpsteinJL18} for this special case. The authors also analyze the \emph{asymptotic approximation ratio} of several algorithms based on \texttt{First-Fit} or \texttt{Next-Fit}, i.e., the limit of the approximation ratio when one restricts attention to instances with $\OPT\to\infty$.
In particular, they give an algorithm whose asymptotic approximation ratio is at most 1.5604.
In addition, it is known that 
the variant of \texttt{First-Fit} considering items in nondecreasing order of their sizes is a $(1+\sqrt{2}/2)$-approximation algorithm
for MSBP~\cite{LiTZ21}.

Another related problem is \textsc{Min-Weighted Sum Set Cover} (MWSSC), in which a collection of subsets $S_1,\ldots,S_m$ of a ground set $[n]$ of items is given in the input,
together with weights $w_1,\ldots,w_n>0$.
As in MWSBP, a solution is a partition $B_1, B_2,\ldots,B_p$ of $[n]$ and the cost of these batches is given by~\eqref{Phi0}, but the difference is that each batch $B_k\subseteq [n]$  must be the subset of some $S_j$, $j\in[m]$. In other words, the difference with MWSBP is that the feasible batches are described \emph{explicitly} by means of a collection of maximal batches rather than \emph{implicitly} using a knapsack constraint $s(B_k) \leq 1$. 
Any instance of MWSBP can thus be cast as an instance of MWSSC,
although this requires an input of exponential size
(enumeration of all maximal subsets of items of total size at most $1$).
The unweighted version of MWSSC was introduced in~\cite{FeigeLT04}. The authors show that 
a natural greedy algorithm is a 4-approximation algorithm,
and that this performance guarantee is the best possible unless $P=NP$.

\vspace{-0.4em}
\paragraph{Contribution and Outline.} 
Given the practical relevance of MWSBP for real-world applications, we feel that it is important to understand the performance of simple algorithms for this problem,
even though a PTAS exists. Indeed, the PTAS
of~\cite{EpsteinL07a} has an astronomical running time, which prevents its usage in applications.
This paper gives improved lower and upper bounds on the approximation ratio of two simple algorithms for MWSBP; In particular, we obtain the first \emph{simple algorithm}
with an approximation algorithm strictly below 2,
see Table~\ref{tab} for a summary of previous and new results.

The two analyzed algorithms, called \texttt{Knapsack-Batching} (\KB) and \texttt{Weighted First-Fit Increasing Reordered}~(\WFFIR), are introduced in Section~\ref{sec:prelim} alongside with more background on algorithms for \textsc{Bin Packing} and MWSBP. In Section~\ref{sec:KB} we show that $\mathcal{R}(\KB)\in (1.433,1.7]$ and in Section~\ref{sec:WFFIR} we show that $\mathcal{R}(\WFFIR)>1.557$. Further, all the best known lower bounds
are achieved for instances in which $w_i=s_i$ for all items, a situation reminiscent of scheduling problems where we minimize the weighted sum of completion times, and where the worst instances have jobs with equal \emph{smith ratios}; see, e.g.~\cite{schwiegelshohn2011alternative}. It is therefore natural to focus on this regime, and we prove that $\WFFIR$ is a $\nicefrac{(7+\sqrt{37})}{8}$-approximation algorithm on instances with $w_i=s_i$ for all $i\in I$.

\begin{table}[t]\label{tab}
 \caption{Summary of previous and new bounds on the approximation ratio of simple algorithms for MWSBP. The results of this paper and indicated in boldface, together with the reference of the proposition (Pr.) or theorem (Th.) where they are proved. \vspace{-0.6em}}
\resizebox{\textwidth}{!}{%
\begin{tabular}{l @{\hskip 1em} rrrr @{\hskip 1em} rr @{\hskip 1em} l}
\hline
 Algorithm & \multicolumn{4}{c@{\hskip 1em}}{lower bounds} & \multicolumn{2}{c@{\hskip 1em}}{upper bounds} & upper bounds for special\\
           & \multicolumn{2}{c}{previous} & \multicolumn{2}{c}{new} & previous & new\quad ~ & ~\quad cases of MWSBP\\ 
 \hline
 \WNFI  & 2&\cite{DobsonN01} & & & 2~\cite{DobsonN01}  & & 1.618 for $w_i=1, \OPT\to\infty$~\cite{EpsteinJL18}\\
 \WFFI  & 2&\cite{DobsonN01} & & & 2~\cite{DobsonN01}  & & 1.707 for $w_i=1$~\cite{LiTZ21}\\
 \WNFDR & 2&\cite{EpsteinL07a} & & & 2~\cite{EpsteinL07a}  & & \\
 \WFFIR & 1.354&\cite{EpsteinJL18} & \textbf{1.557} & (Pr.~\ref{prop:WFFI_LB}) & 2~\cite{DobsonN01} &  & \textbf{1.636} for $w_i=s_i$ (Th.~\ref{theo:UB_WFFIR})\\
 \KB    & 1.354&\cite{EpsteinJL18} & \textbf{1.433} & (Pr.~\ref{prob:LB_GB}) & 4~\cite{FeigeLT04} & \textbf{1.7} (Th.~\ref{theo17})\\
 \hline
\end{tabular}%
}\vspace{-0.5em}
\end{table}

\vspace{-0.7em}\enlargethispage{1em}

\section{Preliminaries}~\label{sec:prelim} \vspace{-1.6em}

Many well known heuristics have been proposed for the \textsc{Bin Packing} problem. For example,
\texttt{Next-Fit} (\NF) and \texttt{First-Fit} (\FF) consider the list of items in an arbitrary given order,
and assign them sequentially to bins. The two algorithms differ
in their approach to select the bin where the current item is placed. To this end, \NF keeps track of a unique opened bin; the current item is placed in this bin if it fits into it, otherwise another bin is opened and the item is placed into it.
In contrast, \FF first checks whether the current item fits in any of the bins used so far. If it does, the item is put into the bin where it fits that has been opened first, and otherwise a new bin is opened.
It is well known that \FF packs all items in at most $\lfloor \nicefrac{17}{10} \cdot OPT \rfloor$ bins,
where $OPT$ denotes the minimal number of required bins, a result obtained by 
D{\'{o}}sa and Sgall~\cite{DosaS13} after a series of papers that improved an additive term in the performance guarantee.

In MWSBP, items not only have a size but also a weight. It is thus natural to consider the weighted variants \WFFI, \WFFD, \WNFI, \WNFI of \FF and \NF, respectively,
where \texttt{W} stands for \emph{weighted}, and the letter \texttt{I} (resp.\ \texttt{D}) stands for \emph{increasing} 
(resp.\ \emph{decreasing}) and indicates that the items are considered in nondecreasing (resp.\ nonincreasing) order of the ratio $s_i/w_i$. 
Using a simple exchange argument, we can see that every algorithm can be enhanced by reordering the bins it outputs by nonincreasing order of their weights.
We denote the resulting algorithms by adding the suffix \texttt{-R} to their acronym.
While \WNFD, \WFFD and \WFFDR do not have a bounded approximation ratio, even for items of unit weights~\cite{EpsteinJL18}, it is shown in~\cite{DobsonN01}
that the approximation ratio of \WNFI and \WFFI is exactly 2,
and the same holds for \WNFDR~\cite{EpsteinL07a}.
This paper gives improved bounds for \WFFIR in Section~\ref{sec:WFFIR}.

\medskip
Another natural algorithm for MWSBP is the 
\texttt{Knapsack-Batching} (\KB) algorithm, which was introduced in~\cite{DobsonN01} and  can be described as follows:
At iteration $k=1,2,\ldots$, solve a knapsack problem to find the subset of remaining items 
$B_k\subseteq [n]\setminus (B_1 \cup \ldots \cup B_{k-1})$
of maximum weight that fits into a bin (i.e., $s(B_k)\leq 1$).
In fact, \cite{EpsteinL07a} argues that \KB is the direct translation of the greedy algorithm for the \textsc{Min-Sum Set Cover} Problem mentioned in the introduction, so its approximation ratio is at most $4$. In practice one can use a \emph{fully polynomial-time approximation scheme} (FPTAS) to obtain a near-optimal solution of the knapsack problem in polynomial-time at each iteration, which yields a $(4+\varepsilon)$-approximation algorithm for MWSBP. In the next section, we show an improved upper bound of 1.7 for the \KB algorithm (or $1.7+\varepsilon$ if an FPTAS is used for solving the knapsack problems in polynomial-time).

\section{The Knapsack-Batching algorithm}\label{sec:KB}
\vspace{-0.4em}

In this section, we study the \texttt{Knapsack-Batching} (\KB) algorithm.
Throughout this section, we denote by $B_1,\ldots,B_p$ the set of bins returned by \KB, and by $O_1,\ldots,O_q$ the optimal bins, for an arbitrary instance of MWSBP. For notational convenience, we also define $B_{p'}=\emptyset$ for all $p'\geq p+1$ and 
$O_{q'}=\emptyset$ for all $q'\geq q+1$.
Recall that \KB solves a knapsack 
over the remaining subset of items at each iteration, so that for all $k$, $w(B_k)\geq w(B)$ holds for all $B\subseteq [n]\setminus (B_1 \cup \ldots \cup B_{k-1})$ such that $s(B)\leq 1$.

We first prove the following proposition, which shows that a performance guarantee of $\alpha$ can be proved if we show that \KB packs at least as much weight in the first $\alpha k$ bins as \OPT does in only $k$ bins. The proof relies on expressing \OPT and \KB as the area below a certain curve. Then, we show that shrinking the curve corresponding to \KB by a factor $\alpha$ yields an underestimator for the \OPT-curve; see Figure~\ref{fig:shrink}. A similar idea was used in~\cite{FeigeLT04} to bound the approximation ratio of the greedy algorithm for \textsc{Min-Sum Set Cover}, but their bound of $4$ results from shrinking the curve by a factor $2$ along both the $x$-axis and the $y$-axis.

\begin{prop}\label{prop:area}
 Let $\alpha\geq 1$. If for all $k\in[q]$ it holds
 \[ w(B_1) + w(B_2) + \ldots + w(B_{\lfloor \alpha k\rfloor})\geq 
 w(O_1) + w(O_2) + \ldots + w(O_k),\]
 then $\KB \leq \alpha \cdot \OPT$.
 \end{prop}

 \begin{proof}
For all $x\geq 0$, let $f^O(x):=\sum_{j=\lfloor x \rfloor+1}^\infty\,  w(O_j)$. 
Note that $f^O$ is piecewise constant
and satisfies $f^O(x)=\sum_{j=k}^q w(O_j)$ for all $x\in[k-1,k), k\in[q]$ and $f^O(x)=0$ for all $x\geq q$.  
As a result, using the second expression in~\eqref{Phi0},
we can express \OPT as the area below the curve of $f^O(x)$:
\[
 \OPT = \sum_{k=1}^q \sum_{j=k}^q w(O_j) = \int_{0}^\infty\ f^O(x)\, dx.
\]
Similarly, we have $\KB=\int_{0}^\infty\ f^A(x)\, dx$, where 
for all $x\geq 0$ we define
$f^A(x):=\sum_{j=\lfloor x \rfloor+1}^\infty w(B_j)$.
The assumption of the lemma can be rewritten as $f^A(\alpha \cdot k)\leq f^O(k)$, for all $k\in [q]$, and we note that this inequality also holds for $k=0$, as
$f^A(0)=f^O(0)=\sum_{i=1}^n w_i$. 

Now, we argue that this implies $f^A(\alpha\cdot x)\leq f^O(x)$, for all $x \geq 0$.
If $x$ lies in an interval of the form $x\in[k,k+1)$ for some $k\in\{0,\ldots,q-1\}$, then we have
$f^O(x)=f^O(k)\geq f^A(\alpha \cdot k)\geq f^A(\alpha \cdot x),$
where the last inequality follows from $k\leq x$ and the fact that $f^A$ is nonincreasing. Otherwise, we have $x\geq q$, so it holds
$f^O(x) = 0 = f^O(q) \geq f^A(\alpha \cdot q) \geq  f^A(\alpha \cdot x)$;
see Figure~\ref{fig:shrink} for an illustration.
We can now conclude this proof:\vspace{-0.1em}
\[
\KB = \int_{0}^\infty\ f^A(x)\, dx 
    = \alpha \cdot \int_{0}^\infty\ f^A(\alpha \cdot y)\, dy
    \leq \alpha \cdot \int_{0}^\infty\ f^O(y)\, dy
    = \alpha \cdot \OPT. \vspace{-1.2em}
\]
\cqfd\vspace{-0.1em}
\end{proof}

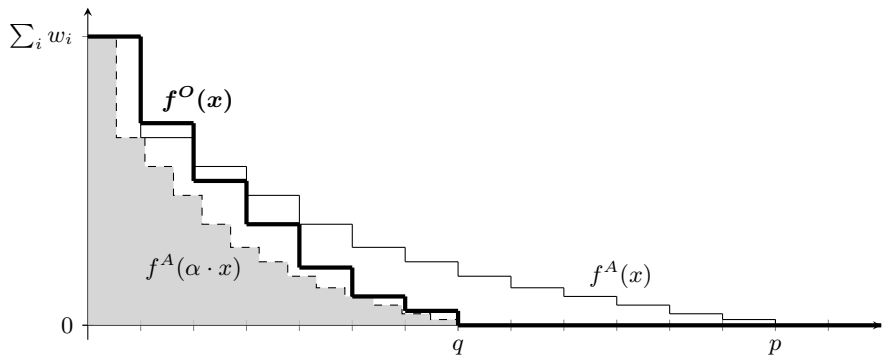
\begin{figure}[t]
\centering
\begin{tikzpicture}
\tikzmath{
\o1 = 100;
\o2 = 70;
\o3 = 50;
\o4 = 35;
\o5 = 20;
\o6 = 10;
\o7 = 5;
\o8 = 0;
\a1 = 100;
\a2 = 65;
\a3 = 55;
\a4 = 45;
\a5 = 35;
\a6 = 27;
\a7 = 22;
\a8 = 17;
\a9 = 13;
\aa0 = 10;
\aa1 = 7;
\aa2 = 4;
\aa3 = 2;
\aa4 = 0;
\al = 0.54;
\x0 = \al*0;
\x1 = \al*1;
\x2 = \al*2;
\x3 = \al*3;
\x4 = \al*4;
\x5 = \al*5;
\x6 = \al*6;
\x7 = \al*7;
\x8 = \al*8;
\x9 = \al*9;
\xx0 = \al*10;
\xx1 = \al*11;
\xx2 = \al*12;
\xx3 = \al*13;
\xx4 = \al*14;
}
\begin{axis}[
axis x line=middle,
axis y line=middle,
xtick={1,...,14},
xticklabels={},
extra x ticks={7,13},
extra x tick labels={$q$,$p$},
xlabel near ticks,
ytick={0.1,100},
yticklabels={$0$,$\sum_i w_i$},
xmax=15,
ymax=110,
xmin=-0.1,
ymin=-5,
width=\textwidth,
height=6cm
]
\tikzstyle{mystyle}=[dashed]
\addplot[domain=0:1] {\a1};
\addplot[domain=1:2] {\a2};
\addplot[domain=2:3] {\a3};
\addplot[domain=3:4] {\a4};
\addplot[domain=4:5] {\a5};
\addplot[domain=5:6] {\a6};
\addplot[domain=6:7] {\a7};
\addplot[domain=7:8] {\a8};
\addplot[domain=8:9] {\a9};
\addplot[domain=9:10] {\aa0} node[above,pos=1] {~$f^A(x)$} ;
\addplot[domain=10:11] {\aa1};
\addplot[domain=11:12] {\aa2};
\addplot[domain=12:13] {\aa3};
\addplot[domain=13:15] {\aa4};
\draw (axis cs:1,\a2) -- (axis cs:1, \a1);
\draw (axis cs:2,\a3) -- (axis cs:2, \a2);
\draw (axis cs:3,\a4) -- (axis cs:3, \a3);
\draw (axis cs:4,\a5) -- (axis cs:4, \a4);
\draw (axis cs:5,\a6) -- (axis cs:5, \a5);
\draw (axis cs:6,\a7) -- (axis cs:6, \a6);
\draw (axis cs:7,\a8) -- (axis cs:7, \a7);
\draw (axis cs:8,\a9) -- (axis cs:8, \a8);
\draw (axis cs:9,\aa0) -- (axis cs:9, \a9);
\draw (axis cs:10,\aa1) -- (axis cs:10, \aa0);
\draw (axis cs:11,\aa2) -- (axis cs:11, \aa1);
\draw (axis cs:12,\aa3) -- (axis cs:12, \aa2);
\draw (axis cs:13,\aa4) -- (axis cs:13, \aa3);
\draw[mystyle, name path=f1] (axis cs:0, \a1) -- (axis cs:\x1, \a1) -- (axis cs:\x1,\a2) -- (axis cs:\x2, \a2) -- (axis cs:\x2,\a3) -- (axis cs:\x3, \a3) -- (axis cs:\x3,\a4) -- (axis cs:\x4, \a4) -- (axis cs:\x4,\a5) -- (axis cs:\x5, \a5) -- (axis cs:\x5,\a6) -- (axis cs:\x6, \a6) -- (axis cs:\x6,\a7) -- (axis cs:\x7, \a7) -- (axis cs:\x7,\a8) -- (axis cs:\x8, \a8)-- (axis cs:\x8,\a9) -- (axis cs:\x9, \a9) -- (axis cs:\x9,\aa0) -- (axis cs:\xx0, \aa0)-- (axis cs:\xx0,\aa1) -- (axis cs:\xx1, \aa1) -- (axis cs:\xx1,\aa2) -- (axis cs:\xx2, \aa2) -- (axis cs:\xx2,\aa3) -- (axis cs:\xx3, \aa3) --(axis cs:\xx3,\aa4);
\path [name path=xaxis]  (axis cs: 13,0) -- (axis cs: 0,0);
 \addplot[black!20, opacity=0.8] fill between [of=xaxis and f1];
\node at (axis cs: 2,20) {$f^A(\alpha \cdot x)$};
\addplot[ultra thick,domain=0:1] {\o1};
\addplot[ultra thick,domain=1:2] {\o2} node[above,pos=1] {~$\boldsymbol{f^O(x)}$};
\addplot[ultra thick,domain=2:3] {\o3};
\addplot[ultra thick,domain=3:4] {\o4};
\addplot[ultra thick,domain=4:5] {\o5};
\addplot[ultra thick,domain=5:6] {\o6};
\addplot[ultra thick,domain=6:7] {\o7};
\addplot[ultra thick,domain=7:15] {\o8};
\draw[ultra thick] (axis cs:1,\o2) -- (axis cs:1, \o1);
\draw[ultra thick] (axis cs:2,\o3) -- (axis cs:2, \o2);
\draw[ultra thick] (axis cs:3,\o4) -- (axis cs:3, \o3);
\draw[ultra thick] (axis cs:4,\o5) -- (axis cs:4, \o4);
\draw[ultra thick] (axis cs:5,\o6) -- (axis cs:5, \o5);
\draw[ultra thick] (axis cs:6,\o7) -- (axis cs:6, \o6);
\draw[ultra thick] (axis cs:7,\o8) -- (axis cs:7, \o7);
\end{axis}
\end{tikzpicture}%
\vspace{-1em}
\caption{\small Illustration of the  proof of Proposition~\ref{prop:area}. The area below the thick curve of $x\mapsto f^O(x)$ is \OPT, and the area under the thin curve of $x\mapsto f^A(x)$ is $\KB$. Shrinking this area horizontally by a factor $\alpha$ produces the shaded area under $x\mapsto f^A(\alpha\cdot x)$, which must be smaller than \OPT. \label{fig:shrink}\vspace{-1em}}
\end{figure}

We next prove that \KB satisfies the assumption of Proposition~\ref{prop:area} for $\alpha=\frac{17}{10}$. This result is of independent interest, as it implies that \KB is also a \nicefrac{17}{10}-approximation algorithm for the problem of finding the smallest number of bins needed to pack a given weight. \vspace{-0.2em}\enlargethispage{1em}

\begin{prop}\label{prop_17}
 For all instances of MWSBP in which the bins $O_1,\ldots, O_q$ are optimal and $\KB$ outputs the bins $B_1,\ldots,B_p$, and for all $k\in[q]$ it holds
 \[ w(B_1) + w(B_2) + \ldots + w(B_{\lfloor \frac{17}{10} k\rfloor})\geq 
 w(O_1) + w(O_2) + \ldots + w(O_k). \vspace{-0.1em}\]
\end{prop}
\begin{proof}
 Let $k\in [q]$. We define the sets
 $\bB:=B_1\cup\ldots\cup B_{\lfloor 1.7k\rfloor}$,
 $\bO:=O_1\cup\ldots\cup O_{k}$. For all $i\in\bB$, denote by $\beta(i)\in[1.7\,k]$ the index of the \KB--bin that contains $i$, i.e., $i\in B_{\beta(i)}$.
 Now, we order the items in $\bO$ in such a way that 
we first consider the items of $\bO\cap\bB$ in nondecreasing order of $\beta(i)$, and then all remaining items in $\bO\setminus\bB$ in an arbitrary order. Now, we construct a new packing $H_1,\ldots,H_{q'}$ of the items in $\bO$, by applying \texttt{First-Fit} to the list of items in $\bO$, considered in the aforementioned order.
For all $i\in \bO\cap\bB$, let $\beta'(i)$ denote the index such that $i\in H_{\beta'(i)}$. Clearly, our order on $\bO$ implies $\beta'(i)\leq \beta(i)$ for all $i\in \bO\cap\bB$.

It follows from \cite{DosaS13} that $q'\leq \lfloor 1.7\cdot k \rfloor$. So we define $H_j:=\emptyset$ for $j=q'+1,$ $q'+2,\ldots,\lfloor 1.7\cdot k \rfloor$, and it holds $\bO = H_1\cup\ldots\cup H_{\lfloor 1.7 k\rfloor}$. Now, we will show that $w(H_j)\leq w(B_j)$ holds for all $j=1,\ldots,\lfloor 1.7k \rfloor$.
To this end, using the greedy property of \KB, it suffices to show that all elements of $H_j$ 
remain when the knapsack problem of the $j$th iteration is solved, i.e.,
$H_j \subseteq [n] \setminus (B_1 \cup \ldots \cup B_{j-1})$. So, let $i\in H_j$. If $i\notin\bB$, then $i\notin (B_1 \cup \ldots \cup B_{j-1})$ is trivial. Otherwise, it is $i\in\bO\cap\bB$, so we have $j=\beta'(i)\leq \beta(i)$, which implies that $i$ does not belong to any $B_\ell$ with $\ell < j$.
We can now conclude the proof:\vspace{-0.2em}
\[
 w(\bO) = \sum_{i=1}^{\lfloor 1.7\cdot k\rfloor} w(H_i)
            \leq \sum_{i=1}^{\lfloor 1.7\cdot k\rfloor} w(B_i)
            = w(\bB).\vspace{-1.8em}
\]
\cqfd\vspace{-0.2em}
\end{proof}

\enlargethispage{1em}
Proposition~\ref{prop:area} and Proposition~\ref{prop_17} yield the main result of this section:
\begin{theo}\label{theo17}
 $\mathcal{R}(\KB) \leq \frac{17}{10}$.
\end{theo}

\noindent\textbf{Remark}\ \emph{It is straightforward to adapt the above proof to show that if we use an FPTAS for obtaining a ($1+\frac{\varepsilon}{n}$)-approximation of the knapsack problems in place of an optimal knapsack at each iteration of the \KB algorithm, we obtain a polynomial-time $(\frac{17}{10}+\varepsilon)$-approximation algorithm for MWSBP.}

\medskip
We next show a lower bound on the approximation ratio of \KB.
For some integers $s$ and $k$ with $s\geq 2^k$, let $\epsilon=\frac{1}{2^k\cdot s}$. Given an integer vector $\nb\in\mathbb{Z}_{\geq 0}^{k+1}$, we construct the following instance:
The items are partitioned in $k+1$ classes, i.e., $[n]=C_1\cup C_2\cup \ldots \cup C_{k+1}$. For all $j\in[k]$, the class $C_j$ consists of $N_j=n_1+\ldots+n_{k+2-j}$ items of class $j$, with $s_i=w_i=\frac{1}{2^j}+\epsilon$, $\forall i\in C_j$. In addition,
the class $C_{k+1}$ contains $N_{k+1}:=n_1 \cdot (s-k)$ tiny items with $s_i=w_i=\epsilon$. We further assume that 
$m_1:=n_1\cdot(2^{-k}-k\epsilon)=\frac{n_1(s-k)}{2^k \cdot s}$ is
an integer, and
for all $j\in[k]$, $\frac{N_j}{2^j-1}$ is an integer. Then, for $j=2,\ldots,k+1$ we let $m_j:=\frac{N_{k+2-j}}{2^{k+2-j}-1}\in\mathbb{Z}$.

\medskip
On this instance, \KB could start by packing $m_1$ bins of total weight $1$, containing $2^k \cdot s$ tiny items each. After
this, there only remains items from the classes $C_1,\ldots,C_k$,
and the solution of the knapsack problem is to put $2^{k}-1$ items of the class $C_k$ in a bin. Therefore, \KB adds $m_2=N_{k}/(2^k-1)$ such bins of weight $(2^k-1)\cdot(2^{-k}+\epsilon)=1-2^{-k}+(2^k-1)\cdot\epsilon$ into the solution. Continuing this reasoning, we see that for each $j=2,\ldots,k+1$, when there only remains items from the classes $C_1,\ldots,C_{k+2-j}$, \KB adds a group of $m_j$ bins that contain $(2^{k+2-j}-1)$ items of class $C_j$, with a weight of $1-2^{-(k+2-j)}+(2^{k+2-j}-1)\cdot\epsilon$ each.
This gives:
\begin{equation}\label{GBi}
 \KB = \sum_{i=1}^{m_1} i + \sum_{j=2}^{k+1} 
 \left(1-2^{-(k+2-j)}+(2^{k+2-j}-1)\cdot \epsilon\right) \cdot \sum_{i=m_1+\ldots+m_{j-1}+1}^{m_1+\ldots+m_{j}} i.
\end{equation}

On the other hand, we can construct the following packing for this instance:
The first $n_1$ bins contain one item of each class $C_1,\ldots,C_k$, plus $(s-k)$ tiny items; their total weight is thus
$\sum_{j=1}^k (2^{-j}+\epsilon) + (s-k)\epsilon = 1-2^{-k} + s\cdot\epsilon=1.$
Then, for each $j=2,\ldots,k+1$, we add a group of $n_j$ bins, each containing one item of each class $C_1,C_2,\ldots,C_{k+2-j}$. The bins in the $j$th group thus contain a weight of $\sum_{i=1}^{k+2-j} (2^{-i}+\epsilon)=1-2^{-(k+2-j)}+(k+2-j)\cdot \epsilon$.
Obviously, the cost of this packing is an upper bound for the optimum:
\begin{equation}\label{OPTi}
 \OPT \leq \sum_{i=1}^{n_1} i + \sum_{j=2}^{k+1} 
 \left(1-2^{-(k+2-j)}+(k+2-j)\cdot \epsilon\right) \cdot \sum_{i=n_1+\ldots+n_{j-1}+1}^{n_1+\ldots+n_{j}} i.\vspace{-0.2em}
\end{equation}
A sketch of these two packings for $k=3$ is shown in Figure~\ref{fig:packings}.

\begin{figure}[t]
\centering
\includegraphics[width=\textwidth]{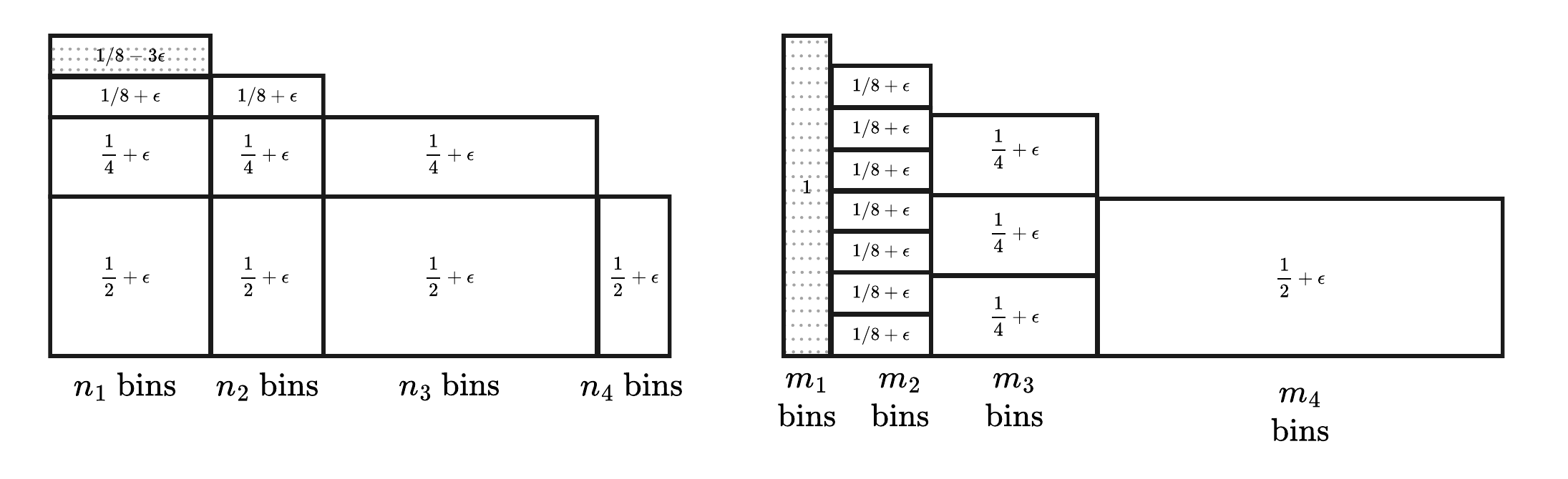}\vspace{-1.5em}
\caption{\small Sketch of the ``good'' packing (left) and the packing returned by \KB (right) for the instance defined before Lemma~\ref{quad}, for $k=3$. The dotted blocks represent bunches of tiny items. The $m_i$'s are chosen so that the number of items in each class is the same in both packings, i.e., $m_1=n_1\cdot (\nicefrac{1}{8}-3\epsilon)$, $m_2=\nicefrac{(n_1+n_2)}{7}$, $m_3=\nicefrac{(n_1+n_2+n_3)}{3}$ and $m_4=n_1+n_2+n_3+n_4$. \vspace{-0.5em}
\label{fig:packings}}
\end{figure}

\medskip
The expressions~\eqref{GBi} and~\eqref{OPTi} are not well suited for finding the values of $n_i$ producing the best lower bound on $\mathcal{R}(\KB)$. The next lemma shows how we can instead focus on maximizing a more amenable ratio of two quadratics.\vspace{-0.2em}\enlargethispage{1em}

\begin{lemm}\label{quad}
 Let $\Lb$ be the $(k+1)\times (k+1)$ lower triangular matrix with all elements equal to $1$ in the lower triangle, $\ub:=[\frac{1}{2^{k+1}},\frac{1}{2^{k+1}},\frac{1}{2^{k}},\ldots,\frac{1}{2^2}]\in\R^{k+1}$,
 $\vb:=[\frac{1}{2^{k}},\frac{1}{2^{k}-1},\frac{1}{2^{k-1}-1},\ldots,\frac{1}{2^{2}-1},\frac{1}{2-1}]\in\R^{k+1}$, and let $\Ub:=\Diag(\ub)$ and $\Vb:=\Diag(\vb)$. Then, for all $\xb\in\R_{>0}^{k+1}$, it holds\vspace{-0.2em}
 \[
  \mathcal{R}(\KB) \geq R(\xb):= \frac{\xb^T \Lb^T \Vb \Lb^T \Ub \Lb \Vb \Lb \xb}{\xb \Lb^T \Ub \Lb \xb}.\vspace{-0.2em}
 \]
\end{lemm}
\begin{proof}
Let $\wb\in\R^{k+1}$ and $\ab\in\mathbb{N}^{k+1}$ such that $a_1<a_2<\ldots<a_{k+1}$, with $a_0:=0$ for notational convenience, and define\vspace{-0.2em}
\begin{align*}
F(\ab):=\sum_{j=1}^{k+1} w_j \sum_{i=a_{j-1}+1}^{a_j} i &=\
\sum_{j=1}^{k+1} \frac{w_j}{2} (a_j^2-a_{j-1}^2 +a_j-a_{j-1})\\
&=\ \ab^T \Diag(\delta(\wb)) \ab + \ab^T \delta(\wb),\vspace{-0.2em}
\end{align*}
where $\delta(\wb):=\frac{1}{2}[w_1-w_2, w_2-w_3,\ldots, w_k-w_{k+1},w_{k+1}]^T\in\R^{k+1}$. Now, assume that $\ab = M \lfloor t\xb \rfloor$ for some arbitrary matrix $M$,  $\xb\in\R_{>0}^{k+1}$ and $t>0$, i.e., $\ab= M \bb$ for some $\bb\in\mathbb{N}^{k+1}$ such that $b_i=\lfloor t x_i\rfloor$, for all $i\in[k+1]$. Using the fact that $\ab = t\cdot (M\xb+o_{t\to\infty}(1))$ and that $F$ is quadratic, we obtain an asymptotic approximation by neglecting the linear terms:
\[F(M \lfloor t\xb \rfloor) = t^2 \cdot \left(\xb^T M^T \Diag\big(\delta(\wb)\big) M \xb + o_{t\to\infty}(1) \right).\]

Now, let $\nb = \alpha \lfloor t \xb \rfloor$ for some arbitrary $\xb\in\R_{>0}^{k+1}$, where $\alpha$  is such that $m_1,m_2,\ldots,m_{k+1}$ are all integers, e.g. $\alpha:= s \cdot 2^k \cdot \prod_{j=1}^k (2^{j}-1)$. We observe that our upper bound on \OPT in~\eqref{OPTi} is of the form $F(\ab)$ with $\ab=[n_1,n_1+n_2,\ldots,n_1+\ldots+n_k]^T=L\nb$ and 
$\wb = [1,1-2^{-k}+k\epsilon,\ldots,1-2^{-1}+\epsilon]^T$, with $\delta(\wb)\to\ub$ as $\epsilon\to 0$. Thus, we have shown
$\OPT \leq t^2(\alpha^2 \xb^T \Lb^T \Db_{\epsilon} \Lb \xb +o_{t\to\infty}(1)$), where $\Db_{\epsilon}:=\Diag(\delta(\wb))$ converges to $\Ub$ as $\epsilon\to 0$.

Similarly,~\eqref{GBi} is also of the form $F(\ab)$, but with
$\ab=L\mb$ and the  weights $\wb$ replaced by $\wb':=[1,1-2^{-k}+(2^{k}-1)\epsilon,\ldots,1-2^{-1}+(2^{-1}-1)\epsilon]$.
Using $\boldsymbol{m}\to\Vb\Lb \nb$ and  $\delta(\wb')\to\ub$ as $\epsilon\to 0$ gives 
the asymptotic approximation
$\KB = t^2 (\alpha^2 \xb^T \Lb^T \Vb \Lb^T \Db_{\epsilon}' \Lb \Vb \Lb \xb + o_{t\to\infty}(1))$, where $\Db_{\epsilon}':=\Diag(\delta(\wb'))$ converges to $\Ub$ as $\epsilon\to 0$. 
This implies the existence of a sequence of instances (with 
$t\to\infty$ and $\epsilon \to 0$) such that
$\KB/\OPT$ converges to a value larger than $R(\xb)$.\cqfd\vspace{-0.3em}
\end{proof}

\begin{prop}~\label{prob:LB_GB}
$\mathcal{R}(\KB) > 1.4334$. \vspace{-0.2em}
\end{prop}
\begin{proof}
 This follows from applying Lemma~\ref{quad} to the vector $\xb=[0.97, 0.01, 0.01,$ $0.01, 0.03, 0.07, 0.15, 0.38]\in \R_{>0}^8$. This vector is in fact the optimal solution of the problem
 of maximizing $R(\xb)$ over $\R^8$, rounded after 2 decimal places for the sake of readability.  To find this vector, we reduced the problem of maximizing $R(\xb)$ to an eigenvalue problem, and recovered $\xb$ by applying a reverse transformation to the corresponding eigenvector. Note that optimizing for $k=50$ instead of $k=7$ only improves the bound by $3\cdot 10^{-5}$.\cqfd\vspace{-1.5em}
\end{proof}

\enlargethispage{1em}

\section{First-Fit-Increasing Reordered}\label{sec:WFFIR} \vspace{-0.7em}

In this section, we analyze the \WFFIR algorithm. First recall that \WFFI (without bins reordering) has an approximation ratio of $2$. To see that this bound is tight, consider the instance with 2 items such that $w_1=s_1=\epsilon$ and $w_2=s_2=1$. Obviously, the optimal packing puts the large item in the first bin and the small item in the second bin, so that $\OPT=1+2\epsilon$. On the other hand, since the two items have the same ratio $s_i/w_i$, \WFFI could put the small item in the first bin, which yields $\WFFI=\epsilon +2$. Therefore, $\WFFI/\OPT\to2$ as $\epsilon\to 0$. In this instance however, we see that reordering the bins by decreasing weight solves the issue. 

It is easy to find an instance in which $\WFFIR/\OPT$ approaches $3/2$, though: Let $w_1=s_1=\epsilon$ and $w_2=w_3=s_2=s_3=\frac{1}{2}$. Now, the optimal packing puts items $2$ and $3$ in the first bin, which gives $\OPT=1+2\epsilon$, while $\WFFIR$ could return the bins $B_1=\{1,2\}, B_2=\{3\}$, so that $\WFFIR=\frac{1}{2}+\epsilon+2\cdot \frac{1}{2}=\frac{3}{2}+\epsilon$.

We next show how to obtain a stronger bound on $\mathcal{R}(\WFFIR)$. To this end, we recall that for all $k\in\mathbb{N}$ and $\epsilon>0$ sufficiently small, there exists an instance of \textsc{Bin Packing} with the following properties~\cite{GareyGU72,DosaS13}: There are $n=30k$ items, which can be partitioned in three classes. The items $i=1,\dots,10k$ are \emph{small} and have size $s_i = 1/6+\delta_i\epsilon$; then there are $10k$ \emph{medium} items of size $s_i=1/3+\delta_i\epsilon$ ($i=10k+1\ldots,20k$), and $10k$ \emph{large} items of size $s_i=1/2+\epsilon$ ($i=20k+1,\ldots,30k$). The constants $\delta_i\in\R$ can be positive or negative and are chosen in a way that, if $\FF$ considers the items in the order $1,\ldots,n$, it produces a packing $B_1,\ldots B_{17k}$ in which the first $2k$ bins contain 5 small items, the next $5k$ bins contain 2 medium items and the last $10k$ bins contain a single large item.
On the other hand, there exists a packing of all items into $10k+1$ bins, in which $10k-1$ bins consist of a small, a medium and a large item and have size $1-O(\epsilon)$, and the two remaining bins have size $1/2+O(\epsilon)$ (a small with a medium item, and a large item alone).

We can transform this \textsc{Bin Packing} instance into a MWSBP instance, by letting $w_i=s_i$, for all $i$. This ensures that all items have the same ratio $s_i/w_i$, so we can assume that \WFFIR considers the items in any order we want. In addition, 
we consider two integers $u$ and $v$, and we increase the number of medium items from $10k$ to $10k+2u$ (so the medium items are $i=10k+1,\ldots,20k+2u$) and the number of large items from $10k$ to $10k+2u+v$ (so the large items are $i=20k+2u+1,\ldots,30k+4u+v$). 
The $\delta_i$'s are unchanged for $i=1,\ldots 20k$, and we let $\delta_i=1$ for all additional medium items ($i=20k+1,\ldots,20k+2u$). Then, assuming that \WFFIR considers the items in the order $1,2,\ldots,$ the algorithm
packs $2k$ bins with 5 small items, $5k+u$ bins with 2 medium items and the last $10k+2u+v$ bins with a single large item. On the other hand, we can use the optimal packing of the original instance, and pack the additional items into $2u$ bins of size $5/6+2\epsilon$ containing a medium and a large item, and $v$ bins with a single large item. This amounts to a total of $10k-1$ bins of size $1-O(\epsilon)$, $2u$ bins of size $5/6+O(\epsilon)$ and $v+2$ bins of size $1/2+O(\epsilon)$. In the limit $\epsilon\to 0$, we get \vspace{-0.5em}
\begin{align*}
 \WFFIR &= \frac{5}{6}\sum_{i=1}^{2k} i  + \frac{2}{3}\sum_{i=2k+1}^{7k+u} i + \frac{1}{2}\sum_{i=7k+u+1}^{17k+3u+v} i\\
 &= \scalemath{0.93}{\frac{5}{6}k(2k+1)+\frac{1}{3}(5k+u)(9k+u+1)+\frac{1}{4}(10k+2u+v)(24k+4u+v+1)}
\end{align*}
and\vspace{-0.5em}
\begin{align*}
\OPT &\geq \sum_{i=1}^{10k-1} i  + \frac{5}{6}\sum_{i=10k}^{10k-1+2u} i + \frac{1}{2}\sum_{i=10k+2u}^{10k+2u+v+1} i\\
&=5k(10k-1) + \frac{5}{6}u(20k+2u-1) + \frac{1}{4}(v+2)(20k+4u+v+1). 
\end{align*}

\clearpage
\begin{prop}\label{prop:WFFI_LB}
 $\mathcal{R}(\WFFIR) > 1.5576$.
\end{prop}
\begin{proof}
The bound is obtained by substituting $k=10350,u=11250,v=24000$ in the above expressions. As in the previous section, these values can be obtained by reducing the problem of finding the best values of $k,u,v$ to an eigenvalue problem and by scaling-up and rounding the obtained eigenvector.\cqfd
\end{proof}

\enlargethispage{1em}

All bad examples for \WFFIR (and even for all algorithms listed in Table~\ref{tab}) have the property that all items have the same ratio $s_i/w_i$. This makes sense intuitively, as the algorithm does not benefit from sorting 
the items anymore, i.e., items can be presented to the algorithm in an adversarial order. While the only known upper bound on the approximation ratio of \WFFIR is 2 (as \WFFIR can only do better than \WFFI), we believe that $\mathcal{R}(\WFFIR)$ is in fact much closer to our lower bound from Proposition~\ref{prop:WFFI_LB}. To support this claim, we next show an upper bound of approx.\ 1.6354 for instances with $s_i=w_i$.

\begin{theo}\label{theo:UB_WFFIR}
For all MWSBP instances with $s_i=w_i$  for all $i\in[n]$, it holds\vspace{-0.1em}
\[\WFFIR \leq \frac{7+\sqrt{37}}{8}\cdot \OPT.\]
\end{theo}

\begin{proof}
Consider an instance of MWSBP with $w_i=s_i$ for all $i\in[n]$ and
denote by $W_1\geq \ldots\geq W_p$ the weight (or equivalently, the size) of the bins $B_1,\ldots,B_p$ returned by $\WFFIR$. 
We first handle the case in which there is at most one bin with weight $\leq \frac{2}{3}$, in which case we obtain a bound of $\frac32$:

\begin{myclaim}\label{C0}
 If $W_{p-1}>\frac23$, then  $\WFFIR\leq \frac32 \OPT$.
\end{myclaim}

This claim is proved in the appendix. We can thus assume w.l.o.g.\ that there are at least $2$ bins with weight $\leq \frac23$.
Let $r\in[p-1]$ denote the index of the first bin such that $W_r\leq \nicefrac{2}{3}$ and define $s:=p-r\in[p-1]$.
We define the vectors $\yb\in\R^r$ and $\xb\in\R^s$ such that
$y_i=W_i$ ($\forall i\in[r]$) and $x_i=W_{r+i}$ ($\forall i\in [s]$). By construction, we have $1\geq y_1 \geq \ldots \geq y_{r-1} > \nicefrac{2}{3}\geq y_r \geq x_1 \geq \ldots \geq x_s$. 
We also define $x_0:=y_r$ and $x_i:=0$ for all $i>s$ for the sake of simplicity.
Note that $W_i+W_j>1$ must hold for all $i\neq j$, as otherwise the 
\texttt{First-Fit} algorithm would have put the items of bins $i$ and $j$ into a single bin. This implies $x_{s-1}+x_s > 1$, hence $x_{s-1}=\max(x_{s-1},x_s) > \nicefrac{1}{2}$.
With this notation, we have: \vspace{-0.6em}\enlargethispage{1em}
\[
 \WFFIR = A(\yb,\xb):=\sum_{i=1}^{r} i \cdot y_i + \sum_{i=1}^{s} (r+i) \cdot  x_i.\vspace{-0.6em}
\]
Next, we prove a lower bound on \OPT depending on $\xb$ an $\yb$. Observe that among the $\WFFIR$-bins $B_r, B_{r+1},\ldots,B_{r+s}$ with weight$\leq \nicefrac{2}{3}$, at most one of them can contain two items or more (the first of these bins that was
opened, as items placed in later bins have a size --and hence a weight-- strictly larger than $\nicefrac{1}{3}$). Thus, $s$ out of these $s+1$ bins must contain a \emph{single item}, and there is no pair of such single items fitting together in a bin. Therefore, the $s$ single items must be placed in distinct bins of the optimal packing. This implies
\begin{equation}\label{lowerO}
w(O_i)\geq w(B_{r+i}) = x_i,\quad  \text{for all}\ \ i\in[s].
\end{equation}

 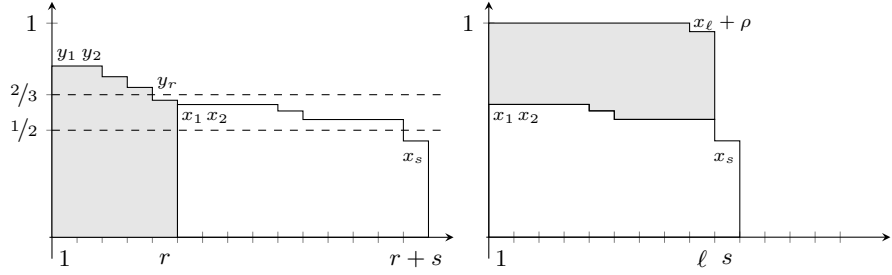
\begin{figure}[t]
\centering
 \begin{tikzpicture}
\tikzmath{
\y1=0.8;
\y2=0.8;
\y3=0.75;
\y4=0.7;
\y5=0.64;
\x1=0.62;
\x4=0.62;
\x5=0.59;
\x6=0.55;
\x9=0.55;
\xx0=0.45;
}
\begin{axis}[
axis x line=middle,
axis y line=middle,
xtick={1,...,14},
xticklabels={},
xlabel near ticks,
ytick={0,0.5,0.666,1},
yticklabels={$0$,$\nicefrac12$,$\nicefrac23$,$1$},
xmax=16,
ymax=1.1,
xmin=-0.1,
ymin=-0.1,
width=0.57\textwidth,
height=5cm,
clip=false
]
\draw[fill=black!10] (axis cs:0,\y1) -- (axis cs:1, \y1) -- (axis cs:1, \y2) -- (axis cs:2, \y2) -- (axis cs:2, \y3) -- (axis cs:3, \y3)-- (axis cs:3, \y4) -- (axis cs:4, \y4) -- (axis cs:4, \y5) -- (axis cs:5, \y5) -- (axis cs:5,0) -- (axis cs:0,0) -- cycle;
\draw[] (axis cs:5,\x1) -- (axis cs:9, \x4) -- (axis cs:9, \x5) -- (axis cs:10, \x5) -- (axis cs:10, \x6) -- (axis cs:14, \x9) -- (axis cs:14, \xx0) -- (axis cs:15,\xx0) -- (axis cs: 15,0) -- (axis cs: 5,0) -- cycle;
\draw[dashed] (axis cs:0,0.666) -- (axis cs:15.5,0.666);
\draw[dashed] (axis cs:0,0.5) -- (axis cs:15.5,0.5);
\node at (axis cs:0.6,\y1+0.05) {\scriptsize $y_1$};
\node at (axis cs:1.6,\y2+0.05) {\scriptsize $y_2$};
\node at (axis cs:4.6,\y5+0.08) {\scriptsize $y_r$};
\node at (axis cs:5.6,\y5-0.08) {\scriptsize $x_1$};
\node at (axis cs:6.6,\y5-0.08) {\scriptsize $x_2$};
\node at (axis cs:14.4,\xx0-0.08) {\scriptsize $x_s$};
\node at (axis cs:0.5,-0.1) {$1$};
\node at (axis cs:4.5,-0.1) {$r$};
\node at (axis cs:14.5,-0.1) {$r+s$};
\end{axis}
\end{tikzpicture}%
\begin{tikzpicture}
\tikzmath{
\x1=0.62;
\x4=0.62;
\x5=0.59;
\x6=0.55;
\x9=0.55;
\xx0=0.45;
\r=0.41;
\z=\x9+\r;
}
\begin{axis}[
axis x line=middle,
axis y line=middle,
xtick={1,...,14},
xticklabels={},
xlabel near ticks,
ytick={0,1},
yticklabels={$0$,$1$},
xmax=16,
ymax=1.1,
xmin=-0.1,
ymin=-0.1,
width=0.57\textwidth,
height=5cm,
clip=false
]
\draw[fill=black!10] (axis cs:0,1) -- (axis cs:8, 1) -- (axis cs:8, \z) -- (axis cs:9, \z) -- (axis cs:9,\x9) -- (axis cs:5,\x9) -- (axis cs:5,\x5) -- (axis cs:4,\x5) -- (axis cs:4,\x4) -- (axis cs:0,\x1) -- cycle;
\draw[] (axis cs:0,\x1) -- (axis cs:4, \x4) -- (axis cs:4, \x5) -- (axis cs:5, \x5) -- (axis cs:5, \x6) -- (axis cs:9, \x9) -- (axis cs:9, \xx0) -- (axis cs:10,\xx0) -- (axis cs: 10,0) -- (axis cs: 0,0) -- cycle;
\node at (axis cs:9.3,1.) {\scriptsize $x_{\ell}+\rho$};
\node at (axis cs:0.5,-0.1) {$1$};
\node at (axis cs:8.5,-0.1) {$\ell$};
\node at (axis cs:9.5,-0.1) {$s$};
\node at (axis cs:0.6,0.64-0.08) {\scriptsize $x_1$};
\node at (axis cs:1.6,0.64-0.08) {\scriptsize $x_2$};
\node at (axis cs:9.4,\xx0-0.08) {\scriptsize $x_s$};
\end{axis}
\end{tikzpicture}%
\vspace{-0.5em}
\caption{\small Sketch of a packing defining $A(\yb,\xb)$ (left), and corresponding pseudo-packing defining $B(\yb,\xb)$ (right): The values of $\ell$ and $\rho$ ensure equality of the two shaded areas. \label{fig:pseudo}\vspace{-2em}}
 \end{figure}

Now, let $\ell\in\mathbb{N}$ be the unique index such that $\sum_{i=1}^{\ell-1} (1-x_i) < \sum_{i=1}^r y_i \leq \sum_{i=1}^{\ell} (1-x_i)$, and define $\rho:=\sum_{i=1}^r y_i-\sum_{i=1}^{\ell-1} (1-x_i)\in(0,1-x_{\ell}]$. We claim that\vspace{-0.3em}
\[
 \OPT \geq B(\yb,\xb):= \frac{1}{2}\ell(\ell-1) + (x_\ell+\rho) \ell + \sum_{i=\ell+1}^s\ i\cdot  x_i,
\]
where the last sum is $0$ if $\ell\geq s$, which corresponds to a ``pseudo-packing'' with weight $1$ in bins $1,\ldots,\ell-1$, weight $x_i$ in the bins $\ell+1,\ldots,s$ and the weight of bin $\ell$ is adjusted to $x_{\ell}+\rho$ so that the total weight equals $\sum_i x_i + \sum_i y_i$, see Figure~\ref{fig:pseudo} for an illustration. This clearly gives a lower bound on $\OPT$, as transferring weight from late bins to early bins only improves the solution, and the pseudo-packing defining $B(\xb,\yb)$ packs the whole weight earlier than any packing  $O_1,\ldots,O_q$ satisfying~\eqref{lowerO}. We extend the definition of the function $B(\cdot,\cdot)$ to any pair of vectors $(\yb',\xb')\in\R^{r}_{>0}\times [0,1]^{t}$ with $t \leq s$, by setting $x'_{i}=0$ for all $i>t$.

The above inequalities implies that $\WFFIR/\OPT$
is bounded from above by
$R(\yb,\xb):=A(\yb,\xb)/B(\yb,\xb)$.
We next give a series of technical claims (proved in the appendix) which allows us to compute the maximum of $R(\yb,\xb)$
when $\xb$ and $\yb$ are as above.
The first claim shows that we obtain an
upper bound for some vectors $\yb'$ and $\xb'$ of the form
$\yb'=[\frac23,\ldots,\frac23,\alpha]\in\R^{r}$ and
$\xb'=[\alpha,\ldots,\alpha]\in\R^{t}$ for some $\alpha\in[\frac{1}{2},\frac{2}{3}]$ and $t\leq s$.
Its proof relies on averaging some coordinates of the vectors $\yb$ and $\xb$, and using a sequential rounding procedure to obtain equal coordinates in the vector $\xb$ and to decrease $y_1,\ldots,y_{r-1}$ to $\frac23$.

\begin{myclaim}\label{C2}
 There exists an $\alpha\in[\frac{1}{2},\frac23]$
 and an integer $t$ such that 
 $R(\yb,\xb)\leq R(\yb',\xb')$ 
 holds for the vectors $\yb'=[\frac23,\ldots,\frac23,\alpha]\in\R^r$ and $\xb'=[\alpha,\ldots,\alpha]\in\R^{t}$.
\end{myclaim}
Next, we give a handy upper bound for $R(\yb',\xb')$ when
$\yb'$ and $\xb'$ are in the form obtained in the previous claim.
\begin{myclaim}\label{C3}
 Let $\yb':=[\frac23,\ldots,\frac23,\alpha]\in\R^r$ and 
 $\xb':=[\alpha,\ldots,\alpha]\in\R^t$. If $\frac23(r-1)+\alpha\leq(1-\alpha)t$, then it holds
\[
 R(\yb',\xb')\leq H_1(\alpha,r,t) := \frac{3(1-\alpha)(3\alpha(t+1)(2r+t)+2r(r-1))}{3\alpha+4r^2+(6\alpha-2)r+9(1-\alpha)\alpha t(t+1)-2}.
\]
Otherwise (if $\frac23(r-1)+\alpha>(1-\alpha)t$), the following bound is valid:
\[
 R(\yb',\xb')\leq H_2(\alpha,r,t) :=\frac{6 r(r-1)+9\alpha(t+1)(2r+t)}{(2 r + 3 \alpha (t + 1) - 2) (2 r + 3 \alpha (t + 1) + 1)}.
\]
\end{myclaim}
We start by bounding the second expression. In that case, we obtain a bound equal to $\frac{13}{8}=1.625$.

\begin{myclaim}\label{C4}
 For all $\alpha\in[\frac12,\frac23]$, $r\geq 1$ and $0\leq t\leq \frac{\nicefrac{2}{3}(r-1)+\alpha}{1-\alpha}$,
 $H_2(\alpha,r,t)\leq\frac{13}{8}$.
\end{myclaim}

Then, we bound $H_1$. It turns out that the derivative of $H(\alpha,r,t)$ with respect to $\alpha$ is nonpositive over the domain of $\alpha$, so we obtain an upper bound by setting $\alpha=\frac12$.
\begin{myclaim}\label{C5}
 For all $\alpha\in[\frac12,\frac23]$, $r\geq 1$ and $t\geq 0$, it holds
 \[
   H_1(\alpha,r,t) \leq 
   H_1(\frac12,r,t) = 
   \frac{3 (4 r^2 + r (6 t + 2) + 3 t (t + 1))}{16 r^2 + 4 r + 9 t^2 + 9 t - 2}.
 \]
\end{myclaim}
Finally, we obtain our upper bound by maximizing the
above expression over the domain $r\geq1, t\geq 0$.
Let $u:=r-1\geq 0$. This allows us to rewrite
the previous upper bound as 
\[H_1(\frac12,r,t)=R_1(u,t):=
\frac{3 (6+3 t^2+10 u+4 u^2+ 9t+6 ut)}{18+9 t+9 t^2+36 u+16 u^2}.\]
for some nonnegative variables $u$ and $t$. Rather than relying on complicated differential calculus to maximize $R_1$, we give a short proof based on a certificate that some matrix is \emph{copositive} (see, e.g.~\cite{dur2010copositive}), that was found by solving a semidefinite program. Observe that $R_1(u,t)=\zb^T A \zb/\zb^T B \zb$, with
\[
\zb=\begin{pmatrix}
     u \\ t \\ 1
    \end{pmatrix},
\quad
 A = \begin{pmatrix}
        12 & 9 & 15\\
        9 & 9 & \nicefrac{27}{2}\\
        15 & \nicefrac{27}{2} & 18
       \end{pmatrix},
\quad \text{and}\quad
B = \begin{pmatrix}
        16 & 0 & 18\\
        0 & 9 & \nicefrac{9}{2}\\
        18 & \nicefrac{9}{2} & 18
       \end{pmatrix}.
\]
Let $\tau:=\frac{7+\sqrt{37}}{8}$,
$X:=\scalemath{0.7}{\begin{pmatrix} 0 & 0 & 3\\ 0 & 0 & \nicefrac{9}{8} \\ 3 & \nicefrac{9}{8} & 0 \end{pmatrix}}$.
The reader can verify that the matrix
\[
 Z=\tau B-A-X = \begin{pmatrix}
             2(1+\sqts) & -9 & \nicefrac{9}{4}\cdot(\sqts-1)\\
             -9         & \nicefrac{9}{8}\cdot(\sqts-1) & \nicefrac{9}{16}\cdot(\sqts-19)\\
             \nicefrac{9}{4}\cdot(\sqts-1) & \nicefrac{9}{16}\cdot(\sqts-19) &  \nicefrac{9}{4}\cdot(\sqts-1)
            \end{pmatrix}
\]
is positive semidefinite. As a result, $\zb^T (\tau B-A) \zb \geq \zb^T X \zb=6u+\nicefrac94\cdot t$ holds for all $u,t\in\R$, and this quantity is nonnegative because $u\geq 0$ and $t\geq 0$. This proves $\tau \cdot \zb^T B \zb - \zb^T A \zb\geq 0$, and thus $R_1(u,t)\leq \tau$.
We note that this bound is tight, as $R_1(u_i,t_i)\to\tau$ for sequences of integers $\{(t_i,u_i)\}_{i\in\mathbb{N}}$ such that $t_i/u_i$ converges to $\frac{2}{9}(1+\sqrt{37})$.
By Claims~\ref{C2}-\ref{C5}, we have thus shown that $R(\yb,\xb)\leq\max(\frac{13}{8},\frac{7+\sqrt{37}}{8})=\frac{7+\sqrt{37}}{8}$, which concludes this proof.
\cqfd
\end{proof}

\noindent\textbf{Acknowledgements.} The author thanks a group of colleagues, in particular Rico Raber and Martin Knaack, for discussions during a group retreat that contributed to this paper. 
The author also thanks anonymous referees, whose comments helped to improve the quality of this manuscript.


\newpage
\appendix

\section*{Appendix: Omitted proofs}

\paragraph{Proof of Claim~\ref{C0}.} 

We assume that the bin weights are $W_1\geq\ldots\geq W_p$, with $W_{p-1}\geq\frac23$. Denote by $S:=\sum_{k} W_k$ the sum of all item weights. Define further $x=\lfloor S \rfloor$ and $f=S-s\in[0,1)$. Clearly, $\OPT \geq \sum_{i=1}^x i + f\cdot (x+1)$, which corresponds to an ideal packing putting the whole weight as early as possible. Then, we write
\[
 \OPT \geq \sum_{i=1}^x i + f\cdot (x+1) = (x+1)\cdot (\frac{x}{2}+f)
 =\frac12 (x+f)(x+f+1)+\frac12 f(1-f),
\]
which implies $\OPT\geq S(S+1)/2$ because $f(1-f)\geq 0$. Then, we get an upper bound on $\WFFIR$ by observing that averaging $W_1,\ldots, W_{p-1}$ can only give a worse packing (as we transfer weight from early bins to later bins):
\[
 \WFFIR\leq \frac{S-W_p}{p-1}\cdot \sum_{i=1}^{p-1} i + p\cdot W_p = \frac{p}{2} (S+W_p).
\]
Furthermore, we know that $S\geq \frac23(p-1) + W_p$, which implies
\[\WFFIR \leq p(S-\frac13(p-1))\leq \frac{(3S+1)^2}{12},\]
where the last inequality is obtained by maximizing over $p$. Then, we obtain
\[
 \frac{\WFFIR}{\OPT}\leq \frac{(3S+1)^2}{6S(S+1)}=\frac32-\frac{3S-1}{6S(S+1)}.
\]
It is straightforward to see that the above bound is $\leq \frac32$ for all $S>1$, while the \WFFIR algorithm is clearly optimal whenever $S\leq 1$ (all items fit in a single bin).

\paragraph{Proof of Claim~\ref{C2}.} 
Let $\ell\in\mathbb{N}$ 
be such that the pseudo-packing defining $B(\xb,\yb)$ uses $\ell$ bins, i.e., there exists $\rho\in(0,1-x_{\ell}]$ such that 
\[\sum_{k=1}^r y_k=\sum_{j=1}^{\ell-1} (1-x_j) + \rho.\]

\enlargethispage{1em}
\sloppypar{
The proof is in two steps. First, we show an upper bound of the form 
\begin{equation}\label{intermediate}
R(\yb,\xb)\leq R([\beta,\ldots,\beta,\alpha],[\alpha,\ldots,\alpha]), 
\end{equation}
and then we will decrease $r-1$ first coordinates of $\yb$ from $\beta$ to $\frac23$. For the first step, we distinguish two cases. We are in the first case
if $\ell>s$ or $\rho\geq \frac13$. Then, we
define
$\beta=\nicefrac{1}{(r-1)}\cdot(y_1+\ldots+y_{r-1})$, 
$\ell_0=\min(\ell,s)$, $\alpha=$\mbox{$\nicefrac{1}{(l_0+1)}\cdot (x_0+\ldots+x_{\ell_0})$}, 
$\bar{\yb}=[\beta,\ldots,\beta,\alpha]\in\R^r$
and $\bar{\xb}=$\mbox{$[\alpha,\ldots,\alpha,x_{\ell_0+1},\ldots,x_s]$} $\in\R^s$, i.e.,
$\bar \yb$ and $\bar \xb$ are obtained by averaging $y_{1},\ldots,y_{r-1}$ and 
$x_{0},\ldots,x_{\ell_0}$, respectively (recall that $x_0=y_r$).
We claim that this averaging of coordinates increases the bound. Indeed, we have $A(\yb,\xb)\leq A(\bar\yb,\bar\xb)$ because $(\bar\yb,\bar\xb)$ is obtained by transferring weight from early bins to later bins. On the other hand, it can be seen that averaging these coordinates did not change the pseudo-packing defining $B$, so that
$B(\yb,\xb)= B(\bar\yb,\bar\xb)$.
This is always true 
when $\ell>s$, and otherwise it requires $\rho\geq \alpha-x_{\ell}$, which is satisfied thanks to $\alpha-x_{\ell}\leq \frac23-\frac13=\frac13$
and our assumption that $\rho\geq \frac13$. Thus, $R(\yb,\xb)\leq R(\bar  \yb,\bar \xb)$.
}

If $\ell>s$, we have already obtained a bound of the form~\eqref{intermediate}. Otherwise, we will use a sequential rounding procedure to round the last coordinates of $\bar \xb$ to either $\alpha$ or to $0$, the latter case being equivalent to deleting some coordinates of $\bar \xb$. Let $j$ be such that $\bar{x_1}\geq\ldots\geq \bar x_{s-j} > \bar x_{s-j+1} = \bar x_{s-j+2} = \ldots = \bar x_s$, i.e., there are $j$ equal coordinates at the end of $\bar \xb$, and denote by $\bar \xb(\delta)$ the vector obtained by adding $\delta\in\R$ to each of the $j$ last coordinates of $\bar \xb$. It is not hard to see that the functions $A(\bar \yb, \bar \xb(\delta))$ and $B(\bar \yb, \bar \xb(\delta))$ are both linear in the interval $[-\bar x_s, \bar x_{s-j}-\bar x_{s-j+1}]$. This implies that 
$R(\yb,\bar \xb(\delta)))$ must be a monotone rational function on this interval, and we obtain a larger bound by setting either $\delta=-\bar{x}_s$, which means that we delete the last $j$ coordinates of $\bar \xb$, or by setting $\delta=\bar{x}_{s-j}-\bar x_{s-j+1}$, in which case we have increased the number of equal coordinates at the end of $\bar\xb$. Repeating this procedure gives a bound of the desired form~\eqref{intermediate}.

In the second case, it is $\ell\leq s$ and $\rho<\frac13$. Then, we proceed roughly as above, but with a slight change. Indeed, now we need to define $\ell_0=\ell-1$ to ensure that averaging $x_0,\ldots,x_{\ell_0}$ does not change the pseudo-packing defining $B(\yb,\xb)$. Then in the last step of the sequential rounding procedure, when we round a group of coordinates $(\ell,\ldots,s')$ to either $0$ or $\alpha$, we must make sure that the number of bins $\ell$ required by the pseudo-packing does not change for all $\delta\in[-x_{\ell},\alpha-x_\ell]$, such that we ``remain in the same piece'' of the piecewise linear function $\delta \mapsto B(\bar \yb, \bar \xb(\delta))$. This is 
true as long as $x_{\ell}+\rho+\delta\leq 1$, which is
guaranteed by our assumption that $\rho<\frac13$, because  $x_\ell+\rho+\delta\leq x_\ell+\rho+\alpha-x_{\ell}=\rho+\alpha<\frac13+\frac23=1$ holds
for all $\rho\leq \alpha-x_{\ell}$.

It remains to prove the second step of this claim, in which we decrease $\beta$ to a value of $\frac{2}{3}$.  There is nothing to show when $r=1$, henceforth we assume $r>1$.
Denote by $\xb'=[\alpha,\ldots,\alpha]\in\R^t$ the vector obtained at the end of the first step of this proof, 
and let $\yb':=[\frac13,\ldots,\frac13,\alpha]\in\R^r$.
Further, define $\ell'\in\mathbb{N}$ as follows: if $\frac{2}{3}(r-1)+\alpha \leq (1-\alpha)\cdot t$, then $\ell'$ is the smallest integer such that 
$\frac{2}{3}(r-1)+\alpha \leq (1-\alpha)\cdot \ell'$; otherwise, it is the smallest integer such that $\frac{2}{3}(r-1)+\alpha \leq (1-\alpha)\cdot t+ (\ell'-t)$. It can be seen that $\ell'$ defines the number of bins with weight $>x_i'$ in the pseudo packing defining $B(\yb',\xb')$, i.e., $\ell'$ is the minimum number of bins required to pack the a total weight of $\sum_{i=1}^r y_i'=\frac23(r-1)+\alpha$ in the space remaining over bins with sizes $\xb'$. 

\enlargethispage{1em}
Next, we show that $\ell'\geq \frac{r}{2}$, by distinguishing two cases. If 
$\frac{2}{3}(r-1)+\alpha \leq (1-\alpha)\cdot t$, then \vspace{-0.5em}
\[\ell' \geq \frac{\frac{2}{3}(r-1)+\alpha}{1-\alpha}
\geq \frac{4}{3}(r-1)+1\geq \frac{r}{2},\vspace{-0.2em}
 \]
where the second inequality follows from $\alpha\geq \frac12$ and the last one is valid
for all $r\geq \frac25$. Otherwise, we have
\[
 \ell' \geq \frac23 (r-1)+\alpha(1+t)\geq \frac23 (r-1)+\frac12 \geq \frac{r}{2},
\]
where we used $\alpha\geq\frac12$ and $t\geq 0$ for the second inequality, and the last inequality is valid for all $r\geq 1$.

To conclude this proof, observe that we have
\[A(\yb',\xb')=A(\bar\yb,\xb')-(\beta-\frac23)\sum_{i=1}^{r-1} i= A(\bar\yb,\xb')-(\beta-\frac23) (r-1) \frac{r}{2}.\]
On the other hand, we have removed a total weight of $(\beta-\frac23)\cdot(r-1)$ in bins with index $\geq \ell'$ in the pseudo-packing defining $B(\yb,\xb')$, so
\[B(\yb',\xb')\leq B(\bar\yb,\xb')-(\beta-\frac23) (r-1) \ell'.\]
Now, $A(\bar \yb,\xb')\geq B(\bar \yb,\xb')$ and $\ell'\geq \frac{r}{2}$ imply 
$A(\bar\yb,\xb') \ell' (r-1)(\beta-\frac{2}{3}) \geq B(\bar \yb,\xb') \frac{r}{2} (r-1)(\beta-\frac{2}{3})$. Finally, subtracting $A(\bar \yb,\xb')\cdot B(\bar \yb,\xb')$ from both sides and dividing by $B(\bar\yb,\xb')\cdot ((\beta-\frac23) (r-1) \ell'-B(\bar\yb,\xb'))<0$ yields the desired result:
\[
 R(\bar \yb,\xb')=\frac{A(\bar \yb,\xb')}{B(\bar \yb,\xb')}
                 \leq \frac{A(\bar\yb,\xb')-(\beta-\frac23) (r-1) \frac{r}{2}}{B(\bar\yb,\xb')-(\beta-\frac23) (r-1) \ell'}
                 \leq \frac{A(\yb',\xb')}{B(\yb',\xb')}=R(\yb',\xb').
\]

\paragraph{Proof of Claim~\ref{C3}.}
Let $\yb':=[\frac23,\ldots,\frac23,\alpha]\in\R^r$ and 
 $\xb':=[\alpha,\ldots,\alpha]\in\R^t$. Then,
 we have
 \[A(\yb',\xb')=\frac13 r(r-1)+\frac{\alpha}{2}(t+1)(2r+t).\]
If $\sum_i y_i' = \frac23(r-1)+\alpha\leq(1-\alpha)t$, then we have
 \[B(\yb',\xb')=\frac{\ell(\ell-1)}{2}+(\alpha+\rho)\ell+\frac{\alpha}{2}(t-\ell)(t+\ell+1),\]
 where $\ell\in\mathbb{N}$ and $\rho\in(0,1-\alpha]$ are such that $\frac{2}{3}(r-1)+\alpha = \ell (1-\alpha)+\rho$.
 Let $\tilde \ell:=\frac{\nicefrac{2}{3}(r-1)+\alpha}{1-\alpha}$, so that $\frac{2}{3}(r-1)+\alpha=\tilde \ell (1-\alpha)$, and define 
 \[B'=\frac{\tilde\ell(\tilde\ell+1)}{2} + \frac{\alpha}{2}(t-\tilde\ell)(t+\tilde\ell+1).\]
 Substituting $\tilde \ell$ for $(\ell-1)+\frac{\rho}{1-\alpha}$ in $B'$, we arrive at
 $B(\yb',\xb')-B'=\frac{\rho(1-(\alpha+\rho))}{2(1-\alpha)}\geq 0$ after some simplifications. Thus, we have 
 \begin{align*}
R(\yb',\xb')\leq \frac{A(\yb',\xb')}{B'}
  &=\frac{\frac13 r(r-1)+\frac{\alpha}{2}(t+1)(2r+t)}{\frac{(2r+1)(3\alpha+2r-2)}{18(1-\alpha)}+\frac{\alpha}{2}t(t+1)}\\
  &=\frac{3(1-\alpha)(3\alpha(t+1)(2r+t)+2r(r-1))}{3\alpha+4r^2+(6\alpha-2)r+9(1-\alpha)\alpha t(t+1)-2}.
 \end{align*}
 
In the other case, it is 
$\sum_i y_i' = \frac23(r-1)+\alpha > (1-\alpha)t$ and we proceed in a similar manner: Let $\ell\in\mathbb{N}$, $\rho\in(0,1]$ and $\tilde \ell\in\R$ be such that 
$\frac23(r-1)+\alpha = (1-\alpha)t+(\ell-t-1)+\rho=(1-\alpha)t+(\tilde \ell-t)$, so that
\[B(\yb',\xb')=\frac{\ell(\ell-1)}{2}+\rho\ell,\]
and define $B'':=\tilde \ell (\tilde\ell+1)/2$. Using $\tilde\ell=\ell-(1-\rho)$, we obtain $B(\yb',\xb')-B''=\frac{\rho}{2}(1-\rho)\geq 0$. Finally, we substitute $\tilde\ell$ for $\frac23(r-1)+\alpha(1+t)$ in $B''$, which gives
 \begin{align*}
R(\yb',\xb')\leq \frac{A(\yb',\xb')}{B''}
  &=\frac{\frac13 r(r-1)+\frac{\alpha}{2}(t+1)(2r+t)}{\frac{1}{18} (2 r + 3 \alpha (t + 1) - 2) (2 r + 3 \alpha (t + 1) + 1)}\\
  &=\frac{6 r(r-1)+9\alpha(t+1)(2r+t)}{(2 r + 3 \alpha (t + 1) - 2) (2 r + 3 \alpha (t + 1) + 1)}.
 \end{align*}

\paragraph{Proof of Claim~\ref{C4}.} 
Straightforward calculus shows that
\begin{align*}
 \frac{dH_2(\alpha,r,t)}{d\alpha} = \frac{-9(t+1)}{h_2(\alpha,r,t)^2} &\Big( r^2 \big[6(2\alpha-1)+4t(3\alpha-1)\big] + \big[t(2+9\alpha^2(1+t)^2)\big] \\
 &\!\ + r\big[18\alpha^2t^2 + 2(18\alpha^2-6\alpha+1) +2(9\alpha^2-6\alpha+3)\big]\Big).
\end{align*}
This expression is nonnegative for all $t\geq 0,r\geq 0$ and $\alpha\geq\frac{1}{2}$, because $2\alpha-1\geq 0$ for $\alpha\geq \frac{1}{2}$ and the two quadratics $18\alpha^2-6\alpha+1$ and $9\alpha^2-6\alpha+3$
are always nonnegative.
Therefore, $H_2(\alpha,r,t)\leq H_2(\frac12,r,t)
=\frac{6 (4 r^2+3 t (1+t)+r (2+6 t))}{(-1+4 r+3 t) (5+4 r+3 t)}$. 
Next, we compute
\[
 \frac{dH_2(\frac12,r,t)}{dt}
 =\frac{18}{h_3(r,t)^2} \Big(t^2[3+6r]+2t[4r^2+10r-5]+[16r^2-2r-5])\Big),
\]
where $h_3(r,t)$ is the denominator of $H_2(\nicefrac{1}{2},r,t)$. This expression is nonnegative for all $t\geq 0$ and $r\geq 1$, because $4r^2+10r-5\geq 0$ and $16r^2-2r-5\geq 0$ hold for all $r\geq 1$. The function $t\mapsto H_2(\nicefrac{1}{2},r,t)$ is thus nondecreasing
in the domain $t\in[0,\nicefrac43 (r-1)+1]$, which implies
\[
 H_2(\alpha,r,t)\leq
 H_2(\frac12,r,\frac43 (r-1)+1)
 = \frac{26 r^2+2r-1}{16r^2+4r-2} 
\]
for all $\alpha\geq\frac{1}{2}$ and 
$t\in[0,\nicefrac43 (r-1)+1]$. Finally, it is straightforward to verify that this upper bound is nondecreasing with respect to $r$ for $r\geq 1$, so we obtain the upper bound of $\frac{13}{8}$ by letting $r\to\infty$.

\paragraph{Proof of Claim~\ref{C5}.} 
Standard algebraic manipulations show that 
\begin{align*}
 \frac{dH_1(\alpha,r,t)}{d\alpha} = 
 \frac{-3}{h(\alpha,r,t)^2} 
  & \Big( 8 r^4 + (t + 1) \big[24 r^3 (2 \alpha - 1) + 3 t (3 \alpha^2 - 4 \alpha + 2)\big] \\
 & \quad + 6 r^2 \big[(t^2 + 3 t) (3 \alpha^2 - 2 \alpha + 1) + 6 \alpha^2 - 4 \alpha + 1\big]\\
 & \quad + 2 r \big[6 t^2 (2 \alpha - 1) + 9 t \alpha^2 + 9 \alpha^2 - 12 \alpha + 5\big]\Big),
 \end{align*}
 where $h_1(\alpha,r,t)$ is the denominator of $H_1(\alpha,r,t)$. It is easy to see that the above expression is nonpositive for all  $r\geq 0, t\geq 0$ and $\alpha\geq\frac{1}{2}$, because 
 $2\alpha-1\geq 0$ for $\alpha\geq\frac{1}{2}$ and
 the quadratic expressions $3\alpha^2-4\alpha+2$,
 $3\alpha^2-2\alpha+1$, $6\alpha^2-4\alpha+1$ and 
 $9\alpha^2-12\alpha+5$ are always nonnegative.
 Therefore, $\alpha\mapsto H_1(\alpha,r,t)$ is nonincreasing on the interval $[\frac12,\frac23]$ and the result follows by substituting $\frac12$ for $\alpha$.
 
\end{document}